\documentclass[10pt,reqno]{amsart}
\usepackage{amsfonts,amssymb,amsmath,amsopn,amsthm,graphicx}
\usepackage{amsxtra, mathrsfs}
\usepackage{colordvi}
\usepackage[usenames,dvipsnames]{color}
\usepackage{amsfonts,amssymb,amsbsy,amsmath,amsthm,dsfont}
\usepackage{mathtools}
\usepackage[colorlinks=true]{hyperref}
 
\usepackage{enumitem}

\usepackage[parfill]{parskip}    % Activate to begin paragraphs with an empty line rather than an indent

\usepackage[parfill]{parskip}    % Activate to begin paragraphs with an empty line rather than an indent

\numberwithin{equation}{section}

\newtheorem{theorem}{Theorem}[section]
\newtheorem{corollary}[theorem]{Corollary}

\newtheorem{proposition}[theorem]{Proposition}

\theoremstyle{definition}

\newtheorem{remark}[theorem]{Remark}
\newtheorem{remarks}[theorem]{Remarks} 
\newtheorem{assumption}[theorem]{Assumption}

\newcommand{\C}{\mathbb{C}}

\renewcommand{\epsilon}{\varepsilon}

\newcommand{\N}{\mathbb{N}}

\renewcommand{\phi}{\varphi}
\newcommand{\R}{\mathbb{R}}

\newcommand{\w}{\mathrm{w}}
\newcommand{\Z}{\mathbb{Z}}

\newcommand{\m}{{\scriptscriptstyle-}}
\newcommand{\pp}{{\scriptscriptstyle+}}
\newcommand{\ppm}{{\scriptscriptstyle\pm}}
\newcommand{\mpp}{{\scriptscriptstyle\mp}}

\author[L.\ Fanelli]{Luca Fanelli}
\address[Luca Fanelli]{
Ikerbasque, Basque Foundation for Science,
48011 Bilbao, Spain,
\newline \phantom{\quad} 
Universidad del Pa\'is Vasco / Euskal Herriko Unibertsitatea,
48080 Bilbao, Spain
\newline \phantom{\quad} \&
BCAM -- Basque Center for Applied Mathematics,
48009 Bilbao, Spain}
\email{\href{mailto:luca.fanelli@ehu.es}{luca.fanelli@ehu.eus}}

\author[Hynek \ Kova\v{r}\'{\i}k]{Hynek Kova\v{r}\'{\i}k}
\address[Hynek Kova\v{r}\'{\i}k]{
DICATAM, Sezione di Matematica
Universit\'a degli studi di Brescia
Via Branze, 38 - 25123
Brescia, ITALY}
\email{\href{mailto:hynek.kovarik@unibs.it}{hynek.kovarik@unibs.it}}


\makeatletter
\@namedef{subjclassname@2020}{%
  \textup{2020} Mathematics Subject Classification}
\makeatother


\begin{document}

\title[Quantitative Hardy inequality for magnetic Hamiltonians ]{Quantitative Hardy inequality for magnetic Hamiltonians }
%\author{Luca Fanelli and Hynek Kova\v{r}\'{\i}k}

\maketitle

\begin{abstract}
In this  paper we present a new method of proof of Hardy type inequalities for two-dimensional quantum Hamiltonians with a magnetic field of finite flux. 
Our approach gives a quantitative lower bound on the 
best constant in these inequalities both for Schr\"odinger and Pauli operators. Pauli  operators with Aharonov-Bohm magnetic field are discussed as well.
\end{abstract}

%%%%%%%%%%%%%%%%%%%%%%%%%%%%%%%%%%%%%%%%%%%%%%%%%%%%%%%%%%%%%%%%%%%%%%%%%%
%%%%%%%%%%%%%%%%%%%%%%%%%%%%%%%%%%%%%%%%%%%%%%%%%%%%%%%%%%%%%%%%%%%%%%%%%%
\section{\bf Introduction} 
The history of Hardy inequalities for magnetic Dirichlet forms goes back to the pioneering works by Laptev and Weidl \cite{lw,weidl, weidl2}. 
The latter show, roughly speaking,  that introducing a magnetic field $B:\R^2\to \R$ turns a critical operator, i.e.~two-dimensional Laplacian, into a subcritical operator,
 i.e.~two-dimensional magnetic Laplacian $(i\nabla +A)^2$, which can be understood as a mathematical manifestation of the diamagnetic effect. Here $A\in L^2_{\rm loc}(\R^2; \R^2)$ is a vector potential such that $\nabla\times A =B$ in the sense of distributions. The fields $A$ and $B$ have the physical interpretation of the magnetic potential and magnetic field, respectively. For magnetic fields with finite 
 normalized flux
 \begin{equation}
\alpha= \frac{1}{2\pi} \int _{\R^2}B\, dx 
\end{equation}
the above phenomenon is expressed by a lower bound on the associated quadratic form in the following way. If $B\neq 0$, then there exists a constant $C(B)>0$ such that 
\begin{equation} \label{laptev-weidl}
\int _{\R^2}  |(i\nabla +A) u|^2 \, dx \, \geq \,  C(B)
\left\{  \begin{array}{ll}
\displaystyle  \int_{\R^2} \frac{|u|^2}{1+|x|^2}\, dx &\ \mbox{if $\alpha\not\in\Z$},\\[18pt]
\displaystyle \int_{\R^2} \frac{|u|^2}{1+|x|^2 \log^2|x|}\, dx & \ \mbox{if $\alpha\in\Z$}
\end{array}
\right.
\end{equation}
for all  $u\in C_0^\infty(\R^2)$. It is important to recall that no lower bound with a non-negative and nonzero integral weight as above holds if $B=0$. 

To prove the above inequality, Laptev and Weidl write the left hand side of \eqref{laptev-weidl} in polar coordinates
(cf.~the first term in \eqref{polar-coord} below), after choosing a suitable gauge: then the integral weight on the right hand side of \eqref{laptev-weidl} arises from
the contribution of the azimuthal part of the magnetic Laplacian. This approach, with suitable modifications, 
was later applied also in settings in which the threshold of the essential spectrum is positive, see \cite{ek}, or in higher dimensions where the magnetic field gives an extra contribution to an 
already existing Hardy weight, \cite{ep, ck}. For further generalizations of \eqref{laptev-weidl} we refer to \cite{bls, et} and to the monograph \cite{bel}.

To the best of our knowledge, the only case in which the constant $C(B)$ 
is known explicitly is the case of an Aharonov-Bohm type magnetic field, namely the following vector potential
\begin{equation*}
A_{\rm ab}(x,y)=\alpha\, \frac{(-y,x)}{x^2+y^2},
\end{equation*}
introduced in \cite{ab} to give a mathematical description of the observed quantum phenomenon 
for which a charged particle can be scattered by an electromagnetic field even if it is confined in a region in which the field is null. In this case, inequality \eqref{laptev-weidl}  takes the 
form
\begin{equation} \label{ab-weight}
\int _{\R^2}  |(i\nabla +A_{\rm ab})\, u|^2 \, dx \, \geq \,   \mu_\alpha^2 \int_{\R^2} \frac{|u|^2}{|x|^2}\, dx \qquad  \forall\, u \in C_0^\infty(\R^2\setminus \{0\}) ,
\end{equation}
cf.~\cite{bal,fklv,lw}, where 
\begin{equation}\label{mu}
\mu_\alpha := \min_{m\in\Z} |m-\alpha| 
\end{equation}
denotes the distance between $\alpha$ and the set of integers.  
Moreover, the constant $ \mu_\alpha^2$ in \eqref{ab-weight} is sharp. 
Notice that the Aharonov-Bohm magnetic field is extremely singular, and the associated vector potential is not locally square integrable. Recently, a similar effect has been showed to hold in \cite{fklv}, for higher-dimensional generalizations of the Aharonov-Bohm field of the form 
\begin{equation*}
A:\R^d\to\R^d,
\qquad
A(x_1,\dots,x_d)=\alpha\, \frac{(-x_2,x_1,0,\dots,0)}{x_1^2+x_2^2}.
\end{equation*}

The purpose of this paper is twofold. First, to find a quantitative lower bound on the constant $C(B)$ in estimate \eqref{laptev-weidl}. Second, to extend the resulting Hardy inequality to the Pauli 
operator
\begin{equation} \label{pauli-operator}
P(A)=  \begin{pmatrix}
H_+(A)& 0 \\
 0  &   H_-(A)
\end{pmatrix}\,  , \qquad H_\pm(A) =  (i\nabla +A)^2  \pm  B
\end{equation}
in  $L^2(\R^2;\C^2)$. The Pauli operator, contrary to the magnetic Schr\"odinger operator, takes into account the interaction between the magnetic filed and the spin of the particle
under consideration, and hence represents a more realistic model from the physical point of view.  The well-known Aharonov-Casher Theorem, see e.g.~\cite{cfks, ev}, implies that if $\alpha\neq 0$, 
then one of the components $H_+(A)$ and $H_-(A)$, depending on the sign of $\alpha$, is subcritical and therefore must satisfy a Hardy type inequality. We remark that a quantitative knowledge of a Hardy inequality is usually fundamental in order to address answers to some basic questions about the spectrum of those Hamiltonians, see e.g. the recent paper \cite{cfk}, related to the Pauli operator.

Although it might seem natural to extend inequality  \eqref{laptev-weidl} to the Pauli operator $P(A)$ following the technique by Laptev and Weidl \cite{lw}, unfortunately this could only work for 
magnetic fields which don't change sign, and more importantly, we would loose the track of the constant $C(B)$.

Instead, we show that proceeding the other way around, i.e.~proving the Hardy inequality for the Pauli operator first and then extending it to the magnetic Laplacian, turns out to be a much more efficient strategy. 
In fact we will show that working directly with the Pauli operator allows us to develop a new method of proof in which the problem is reduced to a family of one-dimensional weighted inequalities. 
With the help of well-known results established in the literature \cite{muck, tom, flw} (see Appendix \ref{sec-app1}), we then calculate the sharp constants in these inequalities. 
The advantage of this approach is that it gives us much more information about the constant in the resulting Hardy inequality. Moreover, it singles out explicitly the distance between the flux and 
the set of integers, which in inequality \eqref{laptev-weidl} appears only implicitly.
 In particular, in the case of non-integer flux the Hardy weight is given by
\begin{equation} \label{weight}
 \mu_\alpha^2\, \frac{\beta(B;\rho)}{\rho^2+|x|^2}\, ,
\end{equation}
where $\rho>0$ is arbitrary, and where $\beta(B;\rho) \leq 1$ depends explicitly on $B$ and $\rho$ (see Theorem \ref{thm-hardy} below for details). 

In Proposition \ref{prop-radial} we calculate the 
constant  $\beta(B;\rho)$ for radial compactly supported magnetic fields.  The optimal value of $\beta(B;\rho)$, 
which is $1$, is never achieved, but a family of magnetic fields can be constructed in such a way that the corresponding constant converges asymptotically to $1$ (see Proposition \ref{prop-sharp}). 

Similar result for integer fluxes is given in Theorem \ref{thm-hardy-integer}. Moreover, the positivity of the Pauli operator allows to carry over the Hardy inequality also to the magnetic Laplacian at the cost 
of multiplying the integral weight by $1/2$ (see Corollary \ref{cor-schr}). 

In the closing Section \ref{sec-ab} we consider the Pauli operator with Aharonov-Bohm magnetic field. It turns out that the Hardy weight 
in this case is the same as for the magnetic Laplacian, cf.~\eqref{ab-weight}, and that the corresponding constant id sharp. This is proved in Theorem \ref{thm-hardy-ab}.

%%%%%%%%%%%%%%%%%%%%%%%%%%%%%%%%%%%%%%%%%%%%%%%%%%%%%%%%%%%%%%%%%%%%%%%%%%
%%%%%%%%%%%%%%%%%%%%%%%%%%%%%%%%%%%%%%%%%%%%%%%%%%%%%%%%%%%%%%%%%%%%%%%%%%
\section{\bf A Hardy inequality for the Pauli operator}
\label{sec-regular}

 Throughout the manuscript we work under the following condition:

\begin{assumption}\label{ass-B}
Assume there exists $p>2$ such that $B\in L^p_{\rm loc}(\R^2)$. Assume in addition that there exists  $\tau >2$ such that 
\begin{equation} \label{B-decay-cond}
|B(x)| \, = \, \mathcal{O}(|x|^{-\tau}) \qquad |x| \to \infty.
\end{equation}
\end{assumption}
We now define
\begin{equation} \label{h-eq}
h(x) =- \frac{1}{2\pi} \int_{\R^2} B(y) \log |x-y|\, dy,
\end{equation}
so that the vector field given by
\begin{equation} \label{A-h-def} 
A_h = (\partial_{x_2} h, - \partial_{x_1} h)
\end{equation}
satisfies $\nabla\times  A_h = B$. Thanks to Assumption \ref{ass-B} we have 
\begin{equation} \label{A-L-infty}
|A_h| \in L^\infty(\R^2).
\end{equation} 
Notice that 
\begin{equation} \label{exp-h-asymp}
e^{h(x)} = |x|^{-\alpha}\big(1+ \mathcal{O}(|x|^{-1})\big) , \qquad |x|\to \infty.  
\end{equation} 
For an arbitrary $\rho\in (0,\infty)$ we define the functions $g_\ppm:\R_+\to \R_+$ by
\begin{equation} \label{f-pm}
g_\pp(\rho\, ;r) =
 \left\{  \begin{array}{ll}
\displaystyle 1 &\ \mbox{if $|x| \leq \rho$},\\[12pt]
\displaystyle \rho^{-\alpha}\, |x|^\alpha & \ \mbox{if $ |x| > \rho $}. 
\end{array}
\right.
, 
 \qquad 
g_\m(\rho\, ; r) = \frac{1}{g_\pp(\rho\, ;r)}\, .
\end{equation}
In view of \eqref{exp-h-asymp} it follows that the constants $k_\ppm(\rho)$ and $K_\ppm(\rho)$ defined by  
\begin{equation} \label{h-bounds}
k_\ppm(\rho) := \inf_{x\in \R^2} \frac{e^{\mpp h(x)}}{g_\ppm(\rho\, ; |x|)} \qquad \text{and}  \qquad K_\ppm(\rho) :=  \sup_{x\in \R^2} \frac{e^{\mpp h(x)}}{g_\ppm(\rho\, ;|x|)}
\end{equation}
depend only on $B$ and satisfy
$$
0<k_\ppm(\rho) \leq K_\ppm(\rho) < \infty \qquad \forall \, \rho >0.
$$
Moreover, let 
\begin{equation} \label{beta}
\beta_\pm(B;\rho) = \left(\frac{k_\ppm(\rho)}{K_\ppm(\rho)}\right)^2\, \in (0,1).
\end{equation}
We now introduce
\begin{equation}
Q_\ppm[u] = \int _{\R^2} \big (\, |(i\nabla +A_h) u|^2 \pm B |u|^2 \big )\, dx, \qquad u\in d(Q_\ppm)
\end{equation}
the quadratic form associated to the components of the Pauli operator in the gauge \eqref{A-h-def}. From Assumption \ref{ass-B}, equation \eqref{A-L-infty} and 
the Sobolev imbedding theorem that the domain of $Q_\ppm$ satisfies
$$
d(Q_\ppm) = H^1(\R^2)\, .
$$
We are now ready to state our first result, concerned with the case of non integer flux.
\smallskip

\begin{theorem} \label{thm-hardy}
Let $B$ satisfy Assumption \ref{ass-B} and suppose that $0<\alpha \not\in\Z$. Let $A\in L^\infty(\R^2;\R^2)$ be such that $\nabla \times A=B$ in the sense of distributions.  Then
\begin{align}
Q_\pp[u] \, &\geq \,  \mu^2_\alpha\, \beta_\pp(B;\rho) \int_{\R^2} \frac{|u|^2}{\rho^2+|x|^2}\, dx \label{main-1} 
\end{align}
holds for all $u\in H^1(\R^2)$ and all $\rho>0$ with $\beta_\pp(B;\rho)$ is defined in \eqref{beta}.
\end{theorem}

\begin{proof}
Owing to the gauge invariance of the Hardy inequality, we may assume that $A=A_h$, cf.~\eqref{A-h-def}.
By density it suffices to prove the claim for any  $u\in C_0^\infty(\R^2)$. 

We write $v = e^{h} u$ so that $v \in C_0^1(\R^2)$, and a  standard calculation shows that
\begin{equation} \label{form-factor}
Q_\pp[u] = \int_{\R^2} e^{-2h} \, |(\partial_{x_1} - i \partial_{x_2}) v |^2\, dx  
\end{equation}
Now let us expand $v$ into its Fourier series 
\begin{equation} \label{v-fourier}
v(x) = \sum_{m\in\Z} e^{im\theta}\, v_m(r),
\end{equation}
where $v_m \in C_0^1[0,\infty)$, and where $r=|x|, \, \theta =\arg(x_1+ix_2)$. Notice that in view of the regularity of $v_m$ we have
\begin{equation} 
v_m(r) = \frac{1}{2\pi} \int_0^{2\pi}\!\! e^{-im\theta}\, v(r,\theta)\, d\theta=  \mathcal{O}(r) \qquad r\to 0, \qquad \forall\, m\neq 0.
\end{equation}
Since 
$$
(\partial_{x_1} - i \partial_{x_2}) i \theta= e^{- i\theta}\, |x|^{-1} \qquad \text{and} \qquad (\partial_{x_1} - i \partial_{x_2}) v_m(r) =  e^{-i\theta}\, v'_m(r) ,
$$
we get, similarly as in \cite[Sec.~10]{weidl}, 
\begin{equation}  \label{eq-timo}
\int_0^{2\pi} |(\partial_{x_1} - i \partial_{x_2}) v |^2\, d\theta=  2\pi  \sum_{m\in\Z}  \Big | v'_m(r) + \frac{m\, v_m(r)}{r} \Big|^2=  2\pi  \sum_{m\in\Z}  r^{-2m}\, \big | \partial_r( r^{m}\, v_m(r))\big|^2\, .
\end{equation} 
Writing $f_m= r^{m} v_m$ we thus conclude with the following lower bound,
\begin{equation} \label{lowerb-1}
Q_\pp[u]\, \geq\, 2\pi k_\pp^2(\rho)\,   \sum_{m\in\Z} \int_0^\infty  g^2_\pp(\rho\, ;r) \,  r^{1-2m}\, |f'_m(r)|^2\, dr\, , 
\end{equation} 
which holds for all $\rho>0$. 
Our next goal is to show that for any fixed $m\in\Z$, 
\begin{equation} \label{hardy-m-fixed}
 \int_0^\infty  g^2_\pp(\rho\, ;r) \,  r^{1-2m}\, |f'_m(r)|^2\, dr\,  \geq c (m,\alpha)  \int_0^\infty  \frac{g^2_\pp(\rho\, ;r) \,  r^{1-2m}}{\delta_{m,0} \rho^2+r^2}\, |f_m(r)|^2\, dr,
\end{equation}
where $\delta_{j,k}$ is the Kronecker delta and where $c(m,\alpha)$ is a constant independent of $\rho$.  We now aim to find an explicit expression for $c(m,\alpha)$ in \eqref{hardy-m-fixed}. 
To do so we will use the weighted one-dimensional Hardy inequalities stated in Appendix \ref{sec-app1} with the weight function
$$
V(t)= g^2_\pp(\rho\, ;t)\, t^{1-2m},
$$ 
see Theorems \ref{thm-class-1}, \ref{thm-class-2} and Corollaries \ref{cor-class-1}, \ref{cor-class-2}. We will distinguish different cases depending on the value of $m$. Note that 
\begin{equation} \label{vanish-infty}
\lim_{t\to \infty} f_m(t)=0 \qquad \forall\, m\in\Z.
\end{equation}

\smallskip

\underline{$0<\alpha <m:$}  in addition to \eqref{vanish-infty}, it holds $\lim_{t\to 0} f_m(t)=0$. We apply Theorem \ref{thm-class-2} with  $W(t) = t^{-2}\, V(t)$. Since
\begin{equation} \label{V-1}
\int_0^s V(t)^{-1}\, dt  =  \rho^{2m}\, \frac{s^{2m}}{2m}\,  \chi_{(0,\rho]}(s) + \rho^{2m}\, \Big( \frac{s^{2m-2\alpha}}{2(m-\alpha)}  -\frac{\alpha}{2m(m-\alpha)}\Big)\, \chi_{(\rho,\infty)}(s) ,
\end{equation}
and
\begin{equation} 
\int_s^\infty W(t)\, dt   = \rho^{-2m} \Big( \frac{s^{-2m}}{2m} + \frac{\alpha}{2m(m-\alpha)}\Big)\chi_{(0,\rho]}(s)  +   \rho^{-2m}  \frac{s^{2\alpha-2m}}{2(m-\alpha)} \, \chi_{(\rho,\infty)}(s),
\end{equation}
we get 
\begin{equation} 
 \sup_{s>0} \Big (\int_0^s V(t)^{-1}\, dt \Big)  \Big (\int_s^\infty W(t)\, dt \Big)  \leq \frac{1}{4(m-\alpha)^2}\, .
\end{equation}
 Altogether we thus deduce from Theorem \ref{thm-class-2} that 
\begin{equation} \label{hardy-m-bigger-alpha}
\int_0^\infty   g^2_\pp(\rho\, ;r)\, \,  r^{1-2m}\, |f'_m(r)|^2\, dr\,  \geq\, (m-\alpha)^2  \int_0^\infty \frac{ g^2_\pp(\rho\, ;r) }{r^2}\, |v_m(r)|^2\, r dr \qquad \forall \, m>\alpha, \ \rho>0.
\end{equation}

\smallskip

\underline{$m= 0$:} \,  we use Theorem \ref{thm-class-1} with $W(t) =   (\rho^2+t^2)^{-1}\, V(t)$. A short calculation gives 
\begin{equation*} 
\int_s^\infty V(t)^{-1}\, dt  =   \Big( \frac{1}{2\alpha} - \log (s/\rho) \Big)\, \chi_{(0,\rho]}(s)  +   \frac{s^{-2\alpha} \rho^{2\alpha}}{2\alpha}\, \chi_{(\rho,\infty)}(s) \, .
\end{equation*}
 Moreover, if $s\leq \rho$, then 
\begin{equation*} 
\int_0^s W(t)\, dt  = \int_0^s \frac{t}{\rho^2+t^2}\, dt \, =  \frac 12 \log\Big(1+\frac{s^2}{\rho^2}\Big)  \leq \frac{\log 2}{2} \qquad s \leq \rho\, .
\end{equation*}
On the other hand, for $s> \rho$ we have 
\begin{align} \label{s>rho} 
\int_0^s W(t)\, dt  & \leq \int_0^\rho \frac{t\, dt}{\rho^2+t^2}+ \rho^{2\alpha}   \int_\rho^s \frac{t^{1-2\alpha}\, dt }{t^2} = \frac{\log 2}{2} +\frac{1}{2\alpha} \big(1-s^{-2\alpha} \rho^{2\alpha}\big)
\qquad s > \rho\, .
\end{align}
Hence
\begin{align*} 
 \sup_{s\leq \rho} \Big (\int_0^s V(t)^{-1}\, dt \Big)  \Big (\int_s^\infty W(t)\, dt \Big)  &= \frac12\,  \sup_{s\leq\rho} \Big(\frac{1}{2\alpha} - \log (s/\rho)\Big)\log\Big(1+\frac{s^2}{\rho^2}\Big)\\
 & =  \frac12\,  \sup_{s\leq 1} \Big(\frac{1}{2\alpha} - \log s\Big)\log(1+s^2) \, \leq\,  \frac12\,  \sup_{s\leq 1} \Big(\frac{1}{2\alpha} - \log s\Big)s^2\\[10pt]
 & = \frac 14 \left\{  \begin{array}{ll}
\displaystyle  \frac{1}{\alpha}  &\ \mbox{if $ \alpha\leq 1$},\\[12pt]
\displaystyle e^{\frac 1\alpha-1}  & \ \mbox{if $ \alpha >1 $}. 
\end{array}
\right.
\, \leq\, \frac{1}{4}\, \max \Big\{1\, , \,  \frac{1}{\alpha^2} \Big\} \, .
\end{align*}
Finally, for  $s> \rho$ we get, in view of \eqref{s>rho}, 
\begin{align*} 
 \sup_{s>\rho} \Big (\int_0^s V(t)^{-1}\, dt \Big)  \Big (\int_s^\infty W(t)\, dt \Big) \,  & \leq\  \sup_{s>\rho}\, \frac{\rho^{2\alpha}s^{-2\alpha}}{4\alpha} \Big(\frac{1}{\alpha} +\log 2 -	\frac{\rho^{2\alpha}s^{-2\alpha}}{\alpha} \Big) \\
 & =  \frac{1}{4\alpha}\,  \sup_{x>1}\,  x\,  \Big(\frac{1}{\alpha} +\log 2 -\frac{x}{\alpha} \Big)   =   \frac{1}{4\alpha^2}\, \max \Big\{ \frac{(1+\alpha \log 2)^2}{4} \, , \, \log 2\Big\} \\[5pt]
 & \leq  \frac{1}{4}\, \max \Big\{1\, , \,  \frac{1}{\alpha^2} \Big\} \, .
\end{align*}
We then deduce from  Theorem \ref{thm-class-1}  that
\begin{equation} \label{hardy-m-zero}
\int_0^\infty  g^2_\pp(\rho\, ;r)\,   |f'_0(r)|^2\, r dr\,  \geq\, \min \big\{ 1, \alpha^2\big\}\,  \int_0^\infty \frac{ g^2_\pp(\rho\, ;r) }{\rho^2+r^2}\  |v_0(r)|^2\, r dr \qquad \forall \ \rho>0.
\end{equation}

\smallskip

\underline{$m < 0$:}  we have
\begin{equation}  \label{V-1-upperb}
\int_s^\infty V(t)^{-1}\, dt  =   \Big( \frac{\alpha\, \rho^{2m}}{2m(\alpha-m)} -\frac{s^{2m}}{2m}\Big)\,   \chi_{(0,\rho]}(s)  +  \frac{\rho^{2\alpha}\,  s^{2m-2\alpha}}{2(\alpha -m)}\, \chi_{(\rho,\infty)}(s) .
\end{equation}
On the other hand, setting $W(t) = t^{-2}\, V(t)$ gives
\begin{equation*} 
\int_0^s W(t)\, dt  = -\frac{s^{-2m}}{2m} \,   \chi_{(0,\rho]}(s)    + \Big( \, \frac{\rho^{-2\alpha} s^{2\alpha-2m}}{2(\alpha-m)} -  \frac{\alpha\, \rho^{2m}}{2m(\alpha-m)}\Big)\, \chi_{(\rho,\infty)}(s) .
\end{equation*}
Hence
\begin{align*} 
 \sup_{s> \rho} \Big (\int_s^\infty V(t)^{-1}\, dt \Big)  \Big (\int_0^s W(t)\, dt \Big)  &=-\frac{1}{4m(\alpha-m)}\,  \leq\,  \frac{1}{4(\alpha-m)}\, \leq \frac 14\, .
 \end{align*}
 and
 \begin{align*} 
 \sup_{0<s\leq  \rho} \Big (\int_s^\infty V(t)^{-1}\, dt \Big)  \Big (\int_0^s W(t)\, dt \Big)  &=  \frac{1}{4m^2}\,  \leq\,  \frac{1}{4}\, .
 \end{align*}
Theorem \ref{thm-class-1} then gives 
\begin{equation} \label{hardy-m-negative}
\int_0^\infty   g^2_\pp(\rho\, ;r)\,   r^{1-2m}\, |f'_m(r)|^2\, dr\,  \geq\,      \int_0^\infty \frac{g^2_\pp(\rho\, ;r)}{r^2}\, |v_m(r)|^2\, r dr \qquad \forall \, m < 0, \ \rho>0.
\end{equation}

\smallskip

\underline{$0<m<\alpha:$} This is the most delicate case since neither Theorem \ref{thm-class-1} nor Theorem \ref{thm-class-2} are applicable with $V$ as above and with $W(t) = t^{-2}\,  V(t) $. Instead, we make use of Corollaries \ref{cor-class-1}  and  \ref{cor-class-2} with the choice $R=\rho$.. Note that for any $0<s<\rho$ we have 
$$
\int_0^s V(t)^{-1}\, dt  =  \frac{s^{2m}}{2m}  \qquad \text{and}\qquad \int_s^\rho W(t)\, dt  = \frac{s^{-2m} -\rho^{-2m}}{2m}. 
$$
Hence 
$$
 \sup_{0<s\leq  \rho} \Big (\int_s^\rho V(t)^{-1}\, dt \Big)  \Big (\int_0^s W(t)\, dt \Big) \, \leq\, \frac{1}{4m^2}\, ,
$$
and since $\lim_{t\to 0} f_m(t)=0,$ Corollary \ref{cor-class-2} implies that
\begin{equation} \label{hardy-m-positive-a}
\int_0^\rho   g^2_\pp(\rho\, ;r)\,   r^{1-2m}\, |f'_m(r)|^2\, dr\,  \geq\,  m^2\int_0^\rho \frac{g^2_\pp(\rho\, ;r)}{r^2}\, |v_m(r)|^2\, r dr \qquad  \forall \rho>0.
\end{equation}
On the other hand, if $s>\rho$, then
$$
\int_s^\infty V(t)^{-1}\, dt  = \rho^{2\alpha}\,  \frac{s^{-2\alpha+2m}}{2(\alpha-m)}  \qquad \text{and}\qquad \int_\rho^s W(t)\, dt  = \rho^{-2\alpha} \, \frac{s^{2\alpha-2m} -\rho^{2\alpha-2m}}{2(\alpha-m)}. 
$$
It follows that 
$$
 \sup_{s> \rho} \Big (\int_s^\infty V(t)^{-1}\, dt \Big)  \Big (\int_\rho^s W(t)\, dt \Big) \, \leq\, \frac{1}{4(m-\alpha)^2}\, ,
$$
Hence by Corollary \ref{cor-class-1} gives
\begin{equation} \label{hardy-m-positive-b}
\int_\rho^\infty   g^2_\pp(\rho\, ;r)\,   r^{1-2m}\, |f'_m(r)|^2\, dr\,  \geq\,  (m-\alpha)^2\int_\rho^\infty \frac{g^2_\pp(\rho\, ;r)}{r^2}\, |v_m(r)|^2\, r dr \qquad  \forall \rho>0.
\end{equation}
This together with \eqref{hardy-m-positive-a} implies 
\begin{equation} \label{hardy-m-positive-c}
\int_0^\infty  g^2_\pp(\rho\, ;r)\,  r^{1-2m}\, |f'_m(r)|^2\, dr\,  \geq\, (m-\alpha)^2 \int_0^\infty \frac{  g^2_\pp(\rho\, ;r) }{r^2}\, |v_m(r)|^2\, r dr \qquad  0<m < \alpha,
\end{equation}
If $\alpha\not\in\Z$, then \eqref{hardy-m-positive-b} combined with \eqref{hardy-m-bigger-alpha}, \eqref{hardy-m-zero} and \eqref{hardy-m-negative} gives 
\begin{align*}
\sum_{m\in\Z} \int_0^\infty  g^2_\pp(\rho\, ;r)  \,  r^{1-2m}\, |f'_m(r)|^2\, dr\,  & \geq\,   \mu^2_\alpha \sum_{m\in \Z} \, \int_0^\infty \frac{g^2_\pp(\rho\, ;r)  }{\rho^2+r^2}\,  |v_m(r)|^2\, r dr .
\end{align*}
The Parseval identity in combination \eqref{h-bounds} then implies that 
\begin{align*}
\sum_{m\in\Z} \int_0^\infty  g^2_\pp(\rho\, ;r) \,  r^{1-2m}\, |f'_m(r)|^2\, dr\,  &\geq   \frac{ \mu^2_\alpha}{2\pi}\,   \int_{\R^2}\frac{g^2_\pp(\rho\, ; |x|)}{\rho^2+|x|^2}\, |v(x)|^2\, dx  \geq 
 \frac{ \mu^2_\alpha}{2\pi\, K_\pp^2(\rho)} \int_{\R^2}\frac{|u(x)|^2 }{\rho^2+|x|^2}\,  dx.
\end{align*}
Inequality \eqref{main-1} thus follows from \eqref{lowerb-1}.
\end{proof}

\begin{remark}
Notice that 
$$
\beta_\pp(B;\rho) = \mathcal{O}(\rho^{2\alpha}) \qquad \rho\to 0 \, .
$$
 In particular $\beta_\pp(B;\rho)\to 0$ as $\rho \to 0$.  This is expected since otherwise by letting $\rho\to 0$ we would obtain, by monotone convergence, a Hardy inequality on $H^1(\R^2)$ with the integral weight proportional to $|x|^{-2}$, which is impossible.
\end{remark} 

\smallskip

\iffalse
\begin{corollary} \label{cor-hardy}
Let $B$ satisfy Assumption \ref{ass-B} and suppose that $0<\alpha \not\in\Z$. Let $A\in L^\infty(\R^2;\R^2)$ be such that $\nabla \times A=B$ in the sense of distributions.  
Let 
$$
C_\pp= \sup_{\rho>0} \beta_\pp(B;\rho) 
$$
Then
\begin{align}
Q_\pp[u] \, &\geq \, C_\pp\,  \mu^2_\alpha \int_{\R^2} \frac{|u|^2}{|x|^2}\, dx \qquad \ ,  \label{main-ort} 
\end{align}
holds for all  $u\in C^\infty_0(\R^2\setminus\{0\})$ satisfying
\begin{equation} \label{cond-ort}
\int_0^{2\pi} u(r,\theta)\, d\theta = 0 \qquad \forall\, r>0.
\end{equation}
\end{corollary}

\begin{proof} 
\texttt{pending}
\end{proof}

\smallskip

\fi

We now pass to the analogous result in the case of an integer flux.

\begin{theorem} \label{thm-hardy-integer}
Let $B$ satisfy Assumption \ref{ass-B} and suppose that $0\leq \alpha \in\Z$. Let $A\in L^\infty(\R^2;\R^2)$ be such that $\nabla \times A=B$ in the sense of distributions. Then the inequality
\begin{align} 
Q_\pp[u] \, & \geq \, \beta_\pp(B;\rho)\,  \frac{\alpha}{4\alpha+\pi}\,  \int_{\R^2} \frac{|u|^2}{\rho^2+|x|^2 \big(1+\log^2( |x|/\rho)\big)}\, dx \ \label{main-2}
\end{align}
holds for all $u\in H^1(\R^2)$ and for all $\rho>0$.
\end{theorem}

\begin{proof}
If $\alpha\in\Z$, then equations \eqref{hardy-m-bigger-alpha},\eqref{hardy-m-zero},\eqref{hardy-m-negative} and \eqref{hardy-m-positive-b} are still valid, but we have to take care of the remaining case $m=\alpha$. 
We have 
\begin{equation}  \label{V-1-alfa}
\int_0^s V(t)^{-1}\, dt  =  \frac{s^{2\alpha}}{2\alpha} \,   \chi_{(0,\rho]}(s)  +\rho^{2\alpha} \Big( \frac{1}{2\alpha}  + \log (s/\rho)\Big)\,  \chi_{(\rho,\infty)}(s),
\end{equation}
As for the weight $W$, we have to include a logarithmic correction and set 
\begin{equation} \label{W-int}
W(t) = \frac{V(t)}{\rho^2+t^2 \big (1+\log^2 (t/\rho)\big)}.
\end{equation}
For $s\leq \rho$ we then get 
\begin{align*} 
\int_s^\infty W(t)\, dt\,  & \leq \int_s^\rho \frac{t^{1-2\alpha}}{\rho^2 +t^2 \big (1+\log^2 (t/\rho)\big)} \, dt \ + \rho^{-2\alpha} \int_\rho^\infty \frac{dt}{ t\, \big (1+\log^2 (t/\rho)\big)} \\[5pt]
& \leq\,  s^{-2\alpha} \int_s^\rho \frac{t}{\rho^2 +t^2 } \, dt +   \rho^{-2\alpha} \int_1^\infty \frac{dx}{ x\, \big (1+\log^2 x\big)}\\[5pt]
& = \frac{ s^{-2\alpha} }{2}\, \log \Big(\frac{2\rho^2}{\rho^2 +s^2}\Big) +\frac\pi 2\, \rho^{-2\alpha}.
\end{align*}
On the other hand, if $s>\rho$, then
\begin{align*} 
\int_s^\infty W(t)\, dt\,  & \leq  \rho^{-2\alpha} \int_s^\infty \frac{dt}{ t\, \big (1+\log^2 (t/\rho)\big)} =  \rho^{-2\alpha}  \Big(\frac \pi 2 -\arctan \big( \log (s/\rho)\big) \Big)\, .
\end{align*}
From  \eqref{V-1-alfa} and a short calculation we thus deduce that 
\begin{align*} 
 \sup_{s\leq\rho} \Big (\int_0^s V(t)^{-1}\, dt \Big)  \Big (\int_s^\infty W(t)\, dt \Big) \,  & \leq\  \frac{\pi +\log 2}{4\alpha}\, ,
 \end{align*}
while
\begin{align*} 
 \sup_{s>\rho} \Big (\int_0^s V(t)^{-1}\, dt \Big)  \Big (\int_s^\infty W(t)\, dt \Big) \,  & =  \sup_{s>\rho}  \Big( \frac{1}{2\alpha}  + \log (s/\rho)\Big) \Big(\frac \pi 2 -\arctan \big( \log (s/\rho)\big) \Big)\\[5pt]
&  \leq\  \frac{\pi}{4\alpha} + \sup_{s>\rho}   \log (s/\rho) \Big(\frac \pi 2 -\arctan \big( \log (s/\rho)\big) \Big)\\[5pt]
& = \frac{\pi}{4\alpha} +1\, .
 \end{align*}
Theorem  \ref{thm-class-2} then implies
\begin{equation*}
\int_0^\infty  g^2_\pp(\rho\, ;r)  \,  r^{1-2\alpha}\, |f'_m (r)|^2\, dr\,  \geq\, \frac{\alpha}{4\alpha+\pi}  \int_0^\infty \frac{ g^2_\pp(\rho\, ;r) }{\rho^2+r^2 \big(1+\log^2( r/\rho)\big)}\, |v_m(r)|^2\, r dr \qquad  m = \alpha. 
\end{equation*}
Since $\alpha \geq 1$, the above estimate combined with \eqref{hardy-m-bigger-alpha},\eqref{hardy-m-zero},\eqref{hardy-m-negative} and \eqref{hardy-m-positive-b}  gives 
\begin{align*}
\sum_{m\in\Z} \int_0^\infty   g^2_\pp(\rho\, ;r)\,   r^{1-2m}\, |f'_m(r)|^2\, dr\,  & \geq \frac{\alpha}{4\alpha+\pi}   \sum_{m\in \Z} \int_0^\infty \frac{  g^2_\pp(\rho\, ;r) }{\rho^2+r^2 (1+ \log^2(r/\rho))}\,   |v_m(r)|^2\, r dr.
\end{align*}
Inequality \eqref{main-2} thus follows in the same way as in the case of non-integer $\alpha$.
\end{proof}

We now state the Hardy inequalities associated to the form $Q_\mu[u]$.

\begin{theorem} \label{thm-hardy-2}
Let $B$ satisfy Assumption \ref{ass-B} and suppose that $\alpha \leq 0$. Then
the inequalities 
\begin{align}
Q_\m[u] \, &\geq   \beta_\m(B;\rho) \,  \mu^2_\alpha \int_{\R^2} \frac{|u|^2}{\rho^2+|x|^2}\, dx \ \qquad \qquad \qquad\qquad\qquad  \text{if} \   \alpha\not\in\Z.  \\[7pt]
Q_\m[u] \, & \geq \,    \beta_\m(B;\rho)\,  \frac{\alpha}{4\alpha+\pi}\,  \int_{\R^2} \frac{|u|^2\, dx}{\rho^2+|x|^2 \big(1+\log^2( |x|/\rho)\big)} \qquad\  \text{if} \  \alpha\in\Z 
\end{align}
hold for all $u\in H^1(\R^2)$ and all $\rho>0$.
\end{theorem}

The proof of Theorem \ref{thm-hardy-2} uses the fact that 
$$
Q_\m[e^h v] = \int_{\R^2} e^{2h} \, |(\partial_{x_1} + i \partial_{x_2}) v |^2\, dx , 
$$
and follows the same line of arguments as the proof of Theorem \ref{thm-hardy}. We omit the details.

Our next task is to state the final version of the hardy inequality for the Pauli operator \eqref{pauli-operator}. To this aim, we introduce the following notation. Let 
\begin{equation} \label{w-pm}
\w_\ppm(\rho\, ;x)=\beta_\ppm(\rho) \left\{  \begin{array}{ll}
\displaystyle  \frac{ \mu^2_\alpha}{\rho^2+|x|^2} &\ \mbox{if $ \alpha\not\in\Z $},\\[15pt]
\displaystyle \frac{\alpha}{4\alpha+\pi}\,  \frac{1}{\rho^2+|x|^2 \big(1+\log^2( |x|/\rho)\big)} &\ \mbox{if $ \alpha\in\Z  $}\, .
\end{array}
\right.
\end{equation}

We have 

\begin{corollary} \label{cor-pauli}
Let $B$ satisfy Assumption \ref{ass-B}  and let $A\in L^\infty(\R^2;\R^2)$ such that $\nabla\times A=B$ in the sense of distributions. Let 
\begin{equation} 
\mathcal W(\rho\, ; x)= \begin{pmatrix}
 \w_\pp(\rho\, ;x) & 0 \\
 0  &   0
\end{pmatrix}\, \quad \text{if} \   \ \alpha\geq 0, \qquad \ \
\mathcal  W(\rho\, ; x)= \begin{pmatrix}
0& 0 \\
 0  &    \w_\m(\rho\, ;x)
\end{pmatrix}\, \quad \text{if} \  \ \alpha \leq 0 .
\end{equation}
Then for any $\rho>0$ the Hardy type inequality
\begin{equation} \label{hardy-operator}
P(A) \geq \mathcal  W(\rho\, ;  \cdot)
\end{equation}
holds in the sense of quadratic forms on $H^1(\R^2;\C^2)$. 
\end{corollary}

\begin{proof}
The claim follows from Theorems \ref{thm-hardy} and \ref{thm-hardy-2}.
\end{proof}

\begin{remarks}
A few comments concerning Corollary \ref{cor-pauli} are in order. 
\begin{enumerate}
%\item Under Assumption \ref{ass-B} it is always possible to construct $A\in L^\infty(\R^2;\R^2)$ such that $\nabla\times A=B$ holds 
 %in the sense of distributions.
 
 \item For $\alpha=0$ we have 
$$
c\,  \leq \, e^{\mpp h(x)} \leq  C \qquad \forall\, x\in\R^2.
$$
for some $0<c<C$, see \eqref{h-bounds}. A standard test function argument then shows that if 
\begin{equation}
Q_\ppm [u] \, \geq \, \int_{\R^2} {\rm w}\,  |u|^2\, dx
\end{equation}
holds for some $\w\geq 0$ and all $u\in H^1(\R^2)$, then $\w=0$. Hence  no Hardy inequality
with a non-trivial integral weight can hold  in the case $\alpha=0$. This is reflected in Corollary \ref{cor-pauli} by the fact that the right hand side of  \eqref{hardy-operator}
vanishes when $\alpha=0$, see \eqref{w-pm}.
 
 \item It is not difficult to verify that under Assumption \ref{ass-B} the decay rate of the Hardy weight $\w_\ppm$ is optimal in the case $\alpha \not\in\Z$. 
 
 \item The optimality of the decay rate of $\w_\ppm$ in the case  $\alpha\in\Z$ is discussed in Proposition \ref{prop-L1} below. 
\end{enumerate}
\end{remarks}

The following result, which  is a generalization of \cite[Lem.~8.1]{kov}, shows that if $B$ satisfies Assumption \ref{ass-B} and if $\alpha\in\Z$, then the logarithmic factor on the right hand side 
of \eqref{main-2} cannot be removed.

\begin{proposition}\label{prop-L1}
Let $B$ satisfy Assumption \ref{ass-B} and assume that $0<\alpha\in\Z$. Suppose moreover that there exists ${\rm w}\geq 0$ such that 
\begin{equation} \label{hardy-L-1}
Q_\pp [u] \, \geq \, \int_{\R^2} {\rm w}\,  |u|^2\, dx
\end{equation}
holds all $u\in H^1(\R^2)$. Then ${\rm w} \in L^1(\R^2)$.
\end{proposition}

\begin{proof}
Without loss of generality we assume that the vector potential is given by the Poincar\'e gauge
\begin{equation} \label{poincare}
A_p(x) = (-x_2, x_1) \int_0^1 B(tx)\,  t\, dt \, .
\end{equation}
In polar coordinates we then have 
 \begin{equation}  \label{polar-coord}
Q_\pp [u]  =  \int_0^\infty\! \!\!\int_0^{2\pi} \Big[ \, |\partial_r u|^2 +r^{-2} |(i\partial_\theta + r a(r,\theta)) u|^2 \Big]\, r\, dr d\theta +  \int_0^\infty\!\! \int_0^{2\pi}  B |u|^2 \, r\, dr d\theta,
\end{equation}
where 
\begin{equation} \label{azi}
a(r,\theta) = \frac 1r \int_0^r B(t,\theta)\, t\, dt
\end{equation}
is the azimuthal component of the vector potential. Let
$$
\Phi(r) = \frac{1}{2\pi} \int_{|x|<r} B(x) \, dx = \frac{r}{2\pi} \int_0^{2\pi} a(r,\theta) d\theta
$$ 
be the normalized flux of $B$ through the disc of radius $r$ centered in the origin, and let 
$$
\psi(\theta,r)= \theta(\alpha-\Phi(r)) -r\int_0^\theta a(r,s) ds \,. 
$$
Clearly, $\psi(0, r)=\psi(2\pi, r)$ for all $r>0$. 
Now, consider the sequence of  $H^1(\R^2)$ test functions defined by
\begin{equation} \label{test-f}
u_n(r,\theta) = e^{ i\psi(\theta,r)} \min \big\{ \big(\log (r n)\big)_\pp ,\,  1,\, \big(\log (e\, n/r)\big)_\pp \big\}.
\end{equation}
A straightforward calculation shows that 
\begin{equation} \label{q-n}
Q_\pp [u_n] =  2\pi\, \int_0^\infty \Big( |u_n'|^2 + r^{-2} (\alpha-\Phi(r))^2 |u_n|^2\Big)\, r\, dr +  \int_0^\infty\!\! \int_0^{2\pi}  B\,  |u_n|^2 \, r\, dr d\theta.
\end{equation} 
It can be easily verified that 
$$
\sup_{n\in\N}\, \int_0^\infty  |u_n'|^2\, r\, dr  < \infty.
$$
Notice moreover that, as $n\to \infty $,  $|u_n| \to 1 $ locally uniformly in $\R^2$, and that, by Assumption \ref{ass-B},
$$
\big | \alpha-\Phi(r) \big | \leq  \frac{1}{2\pi}\int_r^\infty\!\! \int_0^{2\pi}  |  B(t,\theta) |  \, t\, dt d\theta \, \leq \, C \, r^{2-\tau}
$$
holds for some $C>0, \tau >2$, and all $r>0$. Since $B\in L^1(\R^2)$ we conclude from \eqref{hardy-L-1} and \eqref{q-n}, by the dominated convergence theorem, that 
\begin{equation} 
 \int_{\R^2} {\rm w}\,  dx = \lim_{n\to \infty} \int_{\R^2} {\rm w}\,  |u_n|^2\, dx \, \leq \, \limsup_{n\to \infty} \, Q_\pp [u_n] < \infty,
\end{equation}
as claimed.
\end{proof}

%%%%%%%%%%%%%%%%%%%%%%%%%%%%%%%%%%%%%%%%%%

\subsection{\bf Radial magnetic fields with compact support} 
\label{ssec-radial}

We devote this section to study the case of radially symmetric and compactly supported magnetic field, in which we are able to compute explicitly the constant $\beta_+(R)$ in inequality \eqref{main-1}, as we show in the following result.
\begin{proposition} \label{prop-radial}
Let $b\in L^p_{\rm loc}(\R_+; rdr), p>2,$ be a non-negative function supported in the interval $(0, R),\, R>0$. Let
 $B(x) = b(|x|)$ and suppose that $\int_{\R^2} B >0$. Then
 \begin{align}\label{hardy-radial} 
Q_\pp[u] \, &\geq \,  \mu^2_\alpha\, \frac{e^{2\lambda(R)}}{R^{2\alpha}}  \int_{\R^2} \frac{|u|^2}{R^2+|x|^2}\, dx  \qquad \forall\, u\in H^1(\R^2), 
\end{align}
where
\begin{equation} \label{omega}
\lambda(R)= \int_0^R b(t)\, t \log t\, dt.
\end{equation} 
\end{proposition}

\begin{proof}
Since $B$ is radial, $h$ is radial as well, cf.~\eqref{h-eq}. From the Newton's theorem, cf.~the proof of \cite[Thm.~9.7]{LL}, we thus obtain 
\begin{equation} \label{h-radial}
h(r) = -\log r \!\int_0^r b(t)\, t\, dt -\int_r^\infty b(t)\, t\, \log t\, dt.
\end{equation}
So,
$$
e^{-h(r)} = r^{\alpha}   \qquad  \forall\,   r \geq  R.
$$
Moreover, since $b\geq 0$, it follows that $h$ is decreasing. Hence $e^{-h}$ is increasing and therefore $e^{-h(0)} = e^{\lambda(R)} < R^\alpha=e^{-h(R)}$. 
Upon setting $\rho=R$ we then get, cf.~\eqref{h-bounds},
\begin{equation} 
k_\pp(R) = e^{\lambda(R)}  \qquad \text{and} \qquad K_\pp(R) =  R^\alpha\, ,
\end{equation}
and consequently,
 \begin{equation} 
\beta_\pp(R) = \frac{e^{2\lambda(R)}}{R^{2\alpha}}\, .
\end{equation} 
The claim now follows from  inequality \eqref{main-1}.
\end{proof}

\subsubsection*{\bf Example} 
Consider the sequence of radial magnetic fields given by
\begin{equation} \label{example-B}
B_n(x) = b_n(|x|), \qquad 
b_n(r) = \alpha (n+2)\ r^n\, \chi_{(0,1)}(r)\, ,
\end{equation} 
where $0 <\alpha \leq 1/2,$ and where  $\chi_M$ denotes the characteristic function of $M\subseteq \R$.  Then $R=1$, and the total flux is independent of $n$; 
$$
\frac{1}{2\pi} \int_{\R^2} B_n(x)\, dx = \alpha (n+2) \int_0^1 r^{n+1}\, dr= \alpha \qquad \forall \, n\in\N.
$$
On the other hand, 
$$
\lambda_n(1) =  \int_0^1 b_n(t)\, t \log t\, dt = -\frac{\alpha}{n+2}\, .
$$
Since $e^{-x} \geq 1-x$ for all $x\geq0$, this implies that 
 \begin{equation} 
\beta_{\pp}(1)= e^{2\lambda_n(1)} \geq 1 -\frac{2\alpha}{n}\, .
\end{equation} 
Now if we denote
$$
Q_n[u]  :=  \int _{\R^2} \big (\, |(i\nabla +A_n) u|^2 + B_n |u|^2 \big )\, dx, 
$$
with
\begin{equation} \label{A_n} 
A_n= (\partial_{x_2} h_n, - \partial_{x_1} h_n) \ ,\qquad h_n(x) =  \frac{1}{2\pi} \int_{\R^2} B_n(y) \log |x-y|\, dy,
\end{equation}
then equation \eqref{hardy-radial} gives 
 \begin{align}  \label{Qn-lowerb}
Q_n[u] \, &\geq \,  \alpha^2 \Big(1 -\frac{2\alpha}{n}\Big) \, \int_{\R^2} \frac{|u|^2}{1+|x|^2}\, dx  \qquad \forall\, u\in H^1(\R^2), \qquad \forall\,  n\geq 1\, .
\end{align}

\begin{proposition} \label{prop-sharp}
The constant $\alpha^2$ on the right hand side of \eqref{Qn-lowerb} is sharp.
\end{proposition}

\begin{proof}
Since the magnetic field is radial, the gauge \eqref{A-h-def} coincides with the Poincar\'e  gauge \eqref{poincare}.
Hence passing to the polar coordinates we get
$$
Q_n[u] =  \int_0^\infty\! \!\!\int_0^{2\pi} \Big[ \, |\partial_r u|^2 +r^{-2}\, \big | (i\partial_\theta +r a(r)) u\big|^2 \Big]\, r\, dr d\theta +  \int_0^\infty\!\! \int_0^{2\pi}  B_n |u|^2 \, r\, dr d\theta,
$$
where
$$
a(r) = \frac 1 r\int_0^r b(t)\, t \, dt,
$$
cf.~\eqref{azi}.
Now consider the sequence of radial test functions $u_n(x) = u_n(|x|)$ given by
\begin{equation} 
u_n(r) =  \left\{  \begin{array}{ll}
\big(\log (r/n)\big)_\pp  & \text{if} \quad  \  0< r \leq  e n \\[5pt]
1 &   \text{if} \quad  e n \leq  r  \leq  n^2 \\[5pt]
\big(\log (e n^2/r)\big)_\pp   & \text{if} \quad  n^2 \leq  r  
\end{array}
\right.
\end{equation}
Then $u_n\in H^1(\R^2)$ for all $n\geq 1,$ and since the supports of $u_n$ and $b_n$ are disjoint, it follows that 
$$
Q_n[u_n]  = 2\pi \int_0^\infty \Big( \, |\partial_r u_n|^2 +\frac{\alpha^2}{r^2}\,  |u_n|^2 \Big) \, r\, dr \, .
$$
An explicit calculation gives
$$
Q_n[u_n]  = 2\pi \alpha^2 \log n  + \mathcal{O}(1), \qquad \int_{\R^2} \frac{|u_n|^2}{1+|x|^2}\, dx   = 2\pi \log n + \mathcal{O}(1)\, .
$$
Therefore
$$
\lim_{n\to \infty} \frac{Q_n[u_n]}{\displaystyle  \int_{\R^2} \frac{|u_n|^2}{1+|x|^2}\, dx  } = \alpha^2 \, ,
$$ 
which implies the claim.
\end{proof}

%%%%%%%%%%%%%%%%%%%%%%%%%%%%%%%%%%%%%%%%%%%%%%%%%%%%%%%%%%%%%%%%%%%%%%%%%%%%%%%%%%%
\subsection{\bf Hardy inequality for magnetic Schr\"odinger operators}

We now export the above arguments for the Pauli operator to the case of the magnetic Laplacian, obtaining a quantitative Hardy inequality, at the cost 
of multiplying the integral weight by $1/2$. We have the following result.

\begin{corollary} \label{cor-schr}
Let $B$ satisfy Assumption \ref{ass-B}, and  let $A:\R^2\to \R^2$ be  a vector potential associated to $B$. If $\alpha\neq 0$, then
the inequality 
\begin{align*}
 \int_{\R^2} |(i\nabla +A)u|^2 \, dx \, &\geq \, \frac 12 \int_{\R^2} \w_\sigma(\rho\, ; x) |u|^2 \, dx 
\end{align*}
holds for all $u\in H^1(\R^2)$ and all $\rho>0$. Here $\sigma = {\rm sign}(\alpha)$, and $\w_\ppm(\rho)$ is given by \eqref{w-pm}. 
\end{corollary}

\begin{proof}
Let $u\in H^1(\R^2)$.
By Corollary \ref{cor-pauli}, 
\begin{align*}
 \int_{\R^2} |(i\nabla +A)u|^2 \, dx \, &\geq  \int_{\R^2} \w_\sigma(\rho\, ;x) |u|^2 \, dx  \mp {\rm sign}(\alpha)  \int_{\R^2} B\,  |u|^2 \, dx .
\end{align*}
On the other hand, the positivity of the Pauli operator implies that 
\begin{align*}
 \int_{\R^2} |(i\nabla +A)u|^2 \, dx \, &\geq\,   \int_{\R^2} B\,  |u|^2 \, dx  \qquad \text{and} \qquad 
 \int_{\R^2} |(i\nabla +A)u|^2 \, dx \, \geq\,   -\! \int_{\R^2} B\,  |u|^2 \, dx .
\end{align*}
The claim thus follows by summing up the above bounds. 
\end{proof}

%%%%%%%%%%%%%%%%%%%%%%%%%%%%%%%%%%%%%%%%%%%%%%%%%%%%%%%%%%%%%%%%%%%%%%%%%%
%%%%%%%%%%%%%%%%%%%%%%%%%%%%%%%%%%%%%%%%%%%%%%%%%%%%%%%%%%%%%%%%%%%%%%%%%%
\section{\bf Aharonov-Bohm magnetic field}
\label{sec-ab}

In this section we will prove a Hardy inequality for the Pauli operator with an Aharonov-Bohm magnetic field of total (normalized) flux $\alpha$ concentrated in the origin. 
Such a magnetic field is generated by the vector potential 
\begin{equation} 
A_\alpha(x) = \frac{\alpha}{|x|^2}\, (-x_2, x_1).
\end{equation}
Consequently, we have
\begin{equation} \label{h-ab}
h(x) = -\alpha\log |x|,
\end{equation}
where $h$ is defined in \eqref{A-h-def}.
As usual when dealing with the Aharonov-Bohm magnetic field, it is sufficient, by unitary equivalence, to study the case 
\begin{equation} \label{flux-restriction}
\alpha\in (-1,1).
\end{equation}
Since $A_\alpha\not\in L^2_{\rm loc}(\R^2)$, the associated quadratic forms  $Q^\alpha_\ppm$, associated to the spin-up and spin-down component of the Pauli operator $P(A_\alpha)$, 
cannot be defined in the same way as in the case of a regular magnetic field, cf.~Section \ref{sec-regular}. It is well-know that the operator $P(A_\alpha)$, defined on $C_0^\infty(\R^2\setminus\{0\})$,
is not essentially self-adjoint. Various  self-adjoint extensions are considered in \cite{ar,ev,gs,per}.

Here we consider the Friedrichs' extension. Since we are interested in the validity of a Hardy-type inequality, we only have to define the quadratic form associated to the latter, 
namely
\begin{align}
Q^\alpha_\pp[u] & :=  \int_{\R^2}  |x|^{2\alpha} \big|(\partial_{x_1} - i \partial_{x_2}) \big(|x|^{-\alpha} u\big)  \big |^2\, dx 	\\[4pt]
 Q^\alpha_\m[u] &:= \int_{\R^2}  |x|^{-2\alpha} \big|(\partial_{x_1} + i \partial_{x_2}) \big(|x|^{\alpha} u\big)  \big |^2\, dx 
\end{align}
where the form domains are given by  
\begin{equation} \label{u-zero}
 d(Q^\alpha_\pp)  =d(Q^\alpha_\m) =  \big\{ u\in H^1(\R^2): \  A_\alpha  u \in L^2(\R^2) \big\},
\end{equation} 
 see \cite{bcf} for the details.
%with the corresponding form domains
%\begin{equation} \label{f-domains}
%\begin{aligned}
%d(Q^\alpha_\pp) & = \big\{ u\in L^2(\R^2) :\  |x|^{\alpha} (\partial_{x_1} - i \partial_{x_2}) \big(|x|^{-\alpha} u\big) \in   L^2(\R^2) \big\} , 
% \\[5pt]
 %d(Q^\alpha_\m) & = \big\{ u\in L^2(\R^2) :\  |x|^{-\alpha} (\partial_{x_1} + i \partial_{x_2}) \big(|x|^{\alpha} u\big) \in   L^2(\R^2) \big\}.
%\end{aligned}
%\end{equation}
%A short calculation then shows that 
%\begin{align}
%d(Q^\alpha_\pp) & = d(Q^\alpha_\m) = H^1_{ab}(\R^2) :=  \Big\{ u\in H^1(\R^2) :\  \frac{u}{|x|} \in   L^2(\R^2) \Big\} , 
%\end{align}

We have the following result.

\begin{theorem} \label{thm-hardy-ab}
For any  $\alpha\in (-1,1)$ and all $u\in  d(Q^\alpha_\ppm)$ it holds 
\begin{align} \label{hardy-ab}
Q^\alpha_\ppm[u] \, &\geq  \,  \mu^2_\alpha \!\int_{\R^2} \frac{|u|^2}{|x|^2}\ dx .
\end{align}
Recall that $\mu_\alpha$ is given by \eqref{mu}. Moreover, the constant $ \mu^2_\alpha$ is sharp.
\end{theorem}

\begin{proof} 
Inequality \eqref{hardy-ab} is trivially satisfied if $\alpha=0$.
Assume for definiteness that $\alpha \in (0,1)$. 
Let $u\in  d(Q^\alpha_\pp)$ and let $v=|x|^{-\alpha} u$. 
We expand $v$ into its Fourier series as in \eqref{v-fourier}. Owing to  
\eqref{eq-timo} we thus obtain the identity 
\begin{equation}  \label{q+ab}
Q_\pp[u]\, = \, 2\pi   \sum_{m\in\Z} \int_0^\infty    r^{1+2\alpha-2m}\, |f'_m(r)|^2\, dr\, ,
\end{equation}
where  $f_m = r^m\, v_m$. Let 
$$
u_m(r) = \frac{1}{2\pi} \int_0^{2\pi}\!\! e^{-im\theta}\, u(r,\theta)\, d\theta
$$
be the Fourier coefficients of $u$. In view of \eqref{u-zero} and the Parseval identity we then have 
\begin{equation}  \label{lim-um}
\liminf_{r\to 0} u_m(r)=\liminf_{r\to \infty} u_m(r)=0 \qquad \forall\, m\in\Z.
\end{equation} 
Now we mimic the proof of Theorem \ref{thm-hardy} and use Theorems \ref{thm-class-1}, \ref{thm-class-2} with 
\begin{equation} \label{VW-1}
V(t) = t^{2\alpha-2m+1}\qquad \text{and}  \qquad W(t) = t^{2\alpha-2m-1} \, .
\end{equation} 
Note that $v_m(r) = r^{-\alpha} \, u_m(r)$. 

\smallskip

\underline{$0< \alpha <1\leq m:$} by \eqref{lim-um} $\liminf_{r\to 0} f_m(r)=0$. Theorem \ref{thm-class-2}  with \eqref{VW-1} gives 
$$
 \int_0^\infty    r^{1+2\alpha-2m}\, |f'_m(r)|^2\, dr\,  \geq \, (m-\alpha)^2  \int_0^\infty    r^{2\alpha-2m-1}\, |f_m(r)|^2\, dr =  (m-\alpha)^2  \int_0^\infty    r^{2\alpha-1}\, |v_m(r)|^2\, dr\ .
$$

%\smallskip

\underline{$m\leq 0:$}  here we have $\liminf_{r\to \infty} f_m(r)=0$. Theorem \ref{thm-class-1}  with \eqref{VW-1} thus implies 
$$
 \int_0^\infty    r^{1+2\alpha-2m}\, |f'_m(r)|^2\, dr\,  \geq \ \alpha^2  \int_0^\infty    r^{2\alpha-1}\, |v_m(r)|^2\, dr\ .
$$
The last two estimates together with equation \eqref{q+ab} and the Parseval identity prove the claim for $Q_\pp[\, \cdot\, ]$.

\smallskip

Now let $u\in  d(Q^\alpha_\m)$ and $v=|x|^{\alpha} u$. Accordingly, we get $v_m(r) = r^{\alpha}\, u_m(r)$, and 
\begin{equation}  \label{q-ab}
Q_\m[u]\, = \, 2\pi   \sum_{m\in\Z} \int_0^\infty    r^{1-2\alpha+2m}\, |f'_m(r)|^2\, dr\, ,
\end{equation}
where  $f_m = r^{-m}\, v_m$. Hence we put
\begin{equation} \label{VW-2}
V(t) = t^{2m-2\alpha+1}\qquad \text{and}  \qquad W(t) = t^{2m-2\alpha-1} \, , 
\end{equation} 
and proceed as above. 

\smallskip

\underline{$m >0:$} we have $\liminf_{r\to \infty} f_m(r)=0$.  Therefore we apply Theorem \ref{thm-class-1}  with $V,W$ as in  \eqref{VW-2}. This gives  
\begin{equation*} 
 \int_0^\infty    r^{1-2\alpha+2m}\, |f'_m(r)|^2\, dr\,  \geq \,  (m-\alpha)^2  \int_0^\infty    r^{2\alpha-1}\, |v_m(r)|^2\, dr\ .
\end{equation*}

\underline{$m \leq 0:$} here we have  $\liminf_{r\to 0} f_m(r)=0$. Hence we use Theorem  \ref{thm-class-2}  with \eqref{VW-2} and obtain  
\begin{equation*} 
 \int_0^\infty    r^{1-2\alpha+2m}\, |f'_m(r)|^2\, dr\,  \geq \,  \alpha^2  \int_0^\infty    r^{2\alpha-1}\, |v_m(r)|^2\, dr\ .
\end{equation*}
Altogether we arrive again at \eqref{hardy-ab}. 

It remains to prove the sharpness of the constant $ \mu^2_\alpha$ on the right hand side of \eqref{h-ab}. 
Let $k\in\{-1,0,1\}$ be such that 
$$
\mu_\alpha = |\alpha-k| .
$$
We construct a sequence of test functions $u_{n,k}$ given by 
\begin{equation} 
u_{n,k}(x) = e^{ik \theta}  \left\{  \begin{array}{ll}
\displaystyle n^\alpha |x|^\alpha &\ \mbox{if \  $ 0 <|x| < \frac 1n$},\\[7pt]
1  & \  \mbox{if \ $  \frac 1n \leq |x| \leq n $}, \\[7pt]
\displaystyle \big(\log (e\, n/|x|)\big)_\pp & \  \mbox{if \  $  n <  |x| $}. 
\end{array}
\right.
\end{equation}
Then $u_{n,k}\in d(Q^\alpha_\ppm)$ for all $n\in\N$, and an elementary calculation shows that 
\begin{equation} 
\lim_{n\to \infty} \frac{Q_\pp[u_{n,k}]}{\displaystyle \int_{\R^2} \frac{|u_{n,k}|^2}{|x|^2}\, dx }\, = 
\lim_{n\to \infty} \frac{Q_\m[u_{n,-k}]}{\displaystyle \int_{\R^2} \frac{|u_{n,-k}|^2}{|x|^2}\, dx } = (\alpha-k)^2= \mu^2_\alpha\, .
\end{equation} 
This proves the claim for $\alpha >0$.  The proof in the case  $\alpha<0$ is completely analogous. 
\end{proof}

\begin{remark}
Notice that the Hardy inequality \eqref{hardy-ab} holds, with the same integral weight, for both components of the Pauli operator. This is compatible with the 
known fact Pauli operator with the Aharonov-Bohm magnetic field has, contrary to regular magnetic fields, no zero modes, see e.g.~\cite[Thm.~3.1]{ev}, \cite[Sec.~3.2]{per}.
\end{remark}

%%%%%%%%%%%%%%%%%%%%%%%%%%%%%%%%%%%%%%%%%%%%%%%%%%%%%%%%%%%%%%%%%%%%%%%%%%%%%%%%
%%%%%%%%%%%

\appendix

%%%%%%%%%%%%%%%%%%%%%%%%%%%%%%%%%%%%%%%%%%%%%%%%
\section{\bf One-dimensional weighted Hardy inequalities}
\label{sec-app1}

\begin{theorem}  \label{thm-class-1}
Let $V,W$ be nonnegative, a.e.-finite, measurable functions
on $(0,\infty)$ such that
\begin{equation} 
\int_s^\infty V(t)^{-1}\, dt \, < \infty \qquad \forall\, s >0.
\end{equation}
Then for any locally absolutely continuous function $f$ on $(0,\infty)$ with $\liminf_{t\to\infty} |f(t)|=0$ we have 
\begin{equation} 
\int_0^\infty W(t)\, |f(t)|^2\, dt \, \leq\, C(V,W) \, \int_0^\infty V(t)\, |f'(t)|^2\, dt,
\end{equation} 
where the constant $C(V,W)$ satisfies 
\begin{align} 
C(V,W) & \leq 4 \, \sup_{s>0} \Big (\int_s^\infty V(t)^{-1}\, dt \Big)  \Big (\int_0^s W(t)\, dt \Big)  \label{C-upperb-1} 
\end{align} 
\end{theorem} 

\medskip

\begin{corollary}  \label{cor-class-1}
Let $R\in (0, \infty)$, and let $V,W$ be nonnegative, a.e.-finite, measurable functions
on $(0,\infty)$ such that
\begin{equation} 
\int_s^\infty V(t)^{-1}\, dt \, < \infty \qquad \forall\, s >R.
\end{equation}
Then for any locally absolutely continuous function $f$ on $(R,\infty)$ with $\liminf_{t\to\infty} |f(t)|=0$ we have 
\begin{equation} 
\int_R^\infty W(t)\, |f(t)|^2\, dt \, \leq\, C(V,W) \, \int_R^\infty V(t)\, |f'(t)|^2\, dt,
\end{equation} 
where the constant $C(V,W)$ satisfies 
\begin{align} 
C(V,W) & \leq 4 \, \sup_{s>R} \Big (\int_s^\infty V(t)^{-1}\, dt \Big)  \Big (\int_R^s W(t)\, dt \Big)  \label{C-upperb-2} 
\end{align} 
\end{corollary} 

\medskip

\begin{theorem}  \label{thm-class-2}
Let $V,W$ be nonnegative, a.e.-finite, measurable functions
on $(0,\infty)$ such that
\begin{equation} 
\int_0^s V(t)^{-1}\, dt \, < \infty \qquad \forall\, s >0.
\end{equation}
Then for any locally absolutely continuous function $f$ on $(0,\infty)$ with $\liminf_{t\to 0} |f(t)|=0$ we have 
\begin{equation} 
\int_0^\infty W(t)\, |f(t)|^2\, dt \, \leq\, C(V,W) \, \int_0^\infty V(t)\, |f'(t)|^2\, dt,
\end{equation} 
where the constant $C(V,W)$ satisfies 
\begin{align} 
C(V,W) & \leq 4 \, \sup_{s>0} \Big (\int_0^s V(t)^{-1}\, dt \Big)  \Big (\int_s^\infty W(t)\, dt \Big)  \label{C-upperb-3} 
\end{align} 
\end{theorem} 

\medskip

\begin{corollary}  \label{cor-class-2}
Let $R\in (0,\infty)$, and let $V,W$ be nonnegative, a.e.-finite, measurable functions
on $(0,\infty)$ such that
\begin{equation} 
\int_0^s V(t)^{-1}\, dt \, < \infty \qquad \forall\, s <R.
\end{equation}
Then for any locally absolutely continuous function $f$ on $(0,R)$ with $\liminf_{t\to 0} |f(t)|=0$ we have 
\begin{equation} 
\int_0^R W(t)\, |f(t)|^2\, dt \, \leq\, C(V,W) \, \int_0^R V(t)\, |f'(t)|^2\, dt,
\end{equation} 
where the constant $C(V,W)$ satisfies 
\begin{align} 
C(V,W) & \leq 4 \, \sup_{0<s<R} \Big (\int_0^s V(t)^{-1}\, dt \Big)  \Big (\int_s^R W(t)\, dt \Big)  \label{C-upperb-5} 
\end{align} 
\end{corollary} 

The proofs of Theorems \ref{thm-class-1} and \ref{thm-class-2} can be found in \cite{muck, tom}. We refer also to \cite{flw} to the monograph \cite{ok}. 

%\newpage

\section*{Acknowledgements}

L. Fanelli is partially supported 
by the Basque Government through the BERC 2022--2025 program 
and 
by the Spanish Agencia Estatal de Investigación
through BCAM Severo Ochoa excellence accreditation CEX2021-001142-S/MCIN/AEI/10.13039/501100011033.

L. Fanelli is also supported by the projects PID2021-123034NB-I00/MCIN/AEI/10.13039/501100011033 funded by the Agencia Estatal de Investigación, 
IT1615-22 funded by the Basque Government.

\bigskip

\bibliographystyle{amsalpha}

\end{document}